\documentclass{scrartcl}
\usepackage{amssymb}
\usepackage{amsmath}
\usepackage{tikz}

\usepackage[round]{natbib}

\newcommand{\R}{\ensuremath{\mathbb R}}
\providecommand{\ceil}[1]{\ensuremath{\left\lceil #1\right\rceil}}
\providecommand{\ceilgg}[1]{\ensuremath{\biggl\lceil #1\biggr\rceil}}
\providecommand{\ceilg}[1]{\ensuremath{\bigl\lceil #1\bigr\rceil}}

\providecommand{\parencite}{\citep}
\providecommand{\textcite}{\citet}
\newtheorem{proposition}{Proposition}
\providecommand{\toprule}{\hline}
\providecommand{\colrule}{\hline}
\providecommand{\botrule}{\hline}
\providecommand{\tbl}[2]{\begin{center}\small\caption{#1}#2\end{center}}
\providecommand{\crule}[1]{#1}

\makeatletter
\newcommand{\minimize}{minimize}
\newcommand{\subjectto}{subject to}
\newcommand{\minproblem}{\@ifstar\minproblemstar\minproblemplain}
\newcommand{\minproblemplain}[3][]{
  \begin{align}
    \text{#1}\textbf{\minimize}\qquad & #2\\
    \textbf{\subjectto}\qquad & #3
  \end{align}
}
\newcommand{\minproblemstar}[3][]{
  \begin{align*}
    \text{#1}\textbf{\minimize}\qquad & #2\\
    \textbf{\subjectto}\qquad & #3
  \end{align*}
}
\makeatother
\usepackage{hyperref}
\usepackage{doi}
\usepackage{url}

\newcommand{\qed}{\hfill\ensuremath{\blacksquare}}
\newenvironment{proof}{\textbf{Proof.}}{\qed}

\begin{document}


\title{Improved integer programming models for simple assembly line
  balancing and related problems}

\author{Marcus Ritt\\
  Instituto de Inform\'atica\\
  Universidade Federal do Rio Grande do Sul\\
  \texttt{marcus.ritt@inf.ufrgs.br}
  \and Alysson M. Costa\\
  Department of Mathematics and Statistics\\
  University of Melbourne, Parkville 3010, Australia.\\
  \texttt{alysson.costa@unimelb.edu.au} 
}

\maketitle

\begin{abstract}
  We propose a stronger formulation of the precedence constraints and
  the station limits for the simple assembly line balancing
  problem. The linear relaxation of the improved integer program
  theoretically dominates all previous formulations using impulse
  variables, and produces solutions of significantly better quality in
  practice. The improved formulation can be used to strengthen related
  problems with similar restrictions. We demonstrate their
  effectiveness on the U-shaped assembly line balancing problem and on
  the bin packing problem with precedence constraints.

  
\end{abstract}

\section{Introduction}

Let $(N,\leq)$ be a weak partially ordered set of tasks with integral
execution time $t_i$ for $i\in N$. The simple assembly line balancing
problem (SALBP) is to find an assignment $a:N\rightarrow S$ of the
tasks to a linear sequence of stations $S=\{1,2,\ldots,m\}$,
respecting the partial order, i.e.~all tasks $i,j\in N$ with $i\leq j$
satisfy $a(i)\leq a(j)$. The \emph{cycle time} of an assignment is the
largest time needed to execute the tasks assigned to some station. The
problem is said to be of \emph{type 1} (SALBP-1) when the goal is to
minimize the number stations for a given cycle time, and to be of
\emph{type 2} (SALBP-2) when the goal is to minimize the cycle time
for a given number of stations. The decision version of both problems
are NP-complete, since without precedence constraints SALBP-1 reduces
to the bin packing problem, and SALBP-2 to the problem of minimizing
the makespan of a schedule of the tasks on identical parallel
machines.

The SALBP has been extensively studied in the literature, and there
are excellent constructive and heuristic algorithms, as well as exact
solution methods available,
e.g.~\parencite{Scholl.Voss/1996,Fleszar.Hindi/2003,Blum/2008,Scholl.Klein/1999,Sewell.Jacobson/2012}.
A very good overview of the methods can be found in the survey of
\textcite{Scholl.Becker/2006}.

The problem is of interest to researchers, since it forms the core of
a large class of generalized assembly line balancing problems. These
include assembly lines of different layout, e.g.~U-shaped lines, lines
with assignment restrictions, varying task times, or setup
times. \textcite{Becker.Scholl/2006} survey generalized assembly line
balancing problems, and the \textcite{ALB} provides a property-based
search for information on these problems. Results obtained for the
SALBP can often be transferred to generalized problems.

The purpose of this paper is to show that this also applies to integer
programming models for the SALBP. Clearly, mathematical models solved
by standard solvers are not competitive with state-of-the-art
methods. They are nevertheless useful since they frequently serve as
benchmarks for better methods, and are a tool for studying new general
assembly line balancing problems, where such methods are not yet
available. Combined with reduction rules and heuristic solutions
integer programming models solved by standard solvers can be a
reasonable, prototypical solution method.

To obtain the best possible solution and to guarantee a fair
comparison to other methods it is necessary to select the best
model. In this paper we address this problem comparing theoretically
and computationally several models for the SALBP from the literature,
and some improved models proposed in this paper. A survey of models
for the SALBP can be found in \textcite{Baybars/1986} and
\textcite{Scholl/1999}. To the best of our knowledge, no theoretical
comparison of these models has been published before. A computational
study of some models has been provided by \textcite{Pastor.etal/2007}.

We argue that the results obtained for the SALBP can be generalized to
other assembly line balancing problems. This will be demonstrated by
two case studies. We show how the model for the U-shaped assembly line
balancing problem (UALBP-1) proposed by \textcite{Urban/1998} and the
model for the bin packing problem with precedence constraints (BPP-P),
which has been recently introduced by~\textcite{DellAmico.etal/2012}
can be improved by the formulations proposed in this paper. Urban's
model is widely used in the literature as originally
proposed~\parencite{Aase.etal/2003,Agpak.etal/2011,Chiang.etal/2004,Erin/2007,Gokcen.Agpak/2004,Kara.Tekin/2009},
and a better model may improve the results in these applications. In
the case of the BPP-P it turns out that the best integer model is
competitive with a sophisticated tailored branch-and-bound algorithm,
demonstrating its utility as a tool for obtaining a rapid, reasonable
solution method for new problems.

The remainder of this paper is organized as follows. In the next
section we present a formal definition, basic mixed-integer linear
models of the SALBP-1 and the SALBP-2 and further models from the
literature. In Section~\ref{sec:improved} we propose improved
formulations of the precedence constraints and the station limits, and
theoretically compare the resulting models with the existing ones. The
improved formulations are applied to related problems in
Section~\ref{sec:applications}. A computational study is presented in
Section~\ref{sec:computationalstudy} and we offer some conclusions in
Section~\ref{sec:conclusions}.

\section{Integer programming models for the SALBP}
\label{sec:problem}

In this section we review mathematical models for SALBP-1 and SALBP-2
that have been proposed in the literature. For a task $i\in N$, let
$F_i$ denote the set of its \emph{immediate followers}, and $P_i$ the
set of its \emph{immediate predecessors}. Let $S$ be the set of
stations. In the following we suppose that for the SALBP-1 an upper
bound $\overline m$ on the number of stations is known ($|N|$ is such
an upper bound), and that $S=\{1,\ldots,\overline m\}$. For the
SALBP-2 the number of stations $m$ is part of the problem instance. In
this case we set $S=\{1,\ldots,m\}$.

There have been three kinds of models proposed in the literature.
\textcite{Bowman/1960} (in the revised formulation of
\textcite{White/1961}) proposed two formulations, one using binary
\emph{impulse variables} and another based on \emph{time variables},
representing the starting time of the tasks. \textcite{Scholl/1999}
proposed a formulation using binary \emph{step variables} $\overline
x_{si}$, where $\overline x_{si}=1$ indicates that task $i\in N$ is
assigned to station $s\in S$ or some preceding station. Since the
formulation using time variables has been found inferior by
\textcite{Pastor.etal/2007} and the formulations using impulse
variables are the most common in the literature, we focus in the
following on the latter.

\subsection{Basic models for SALB}

To represent the assignment of tasks to station, we introduce impulse
variables
\begin{align}
  x_{si}=
  \begin{cases}
    1 & \text{if task $i\in N$ is assigned to station $s\in S$}\\
    0 & \text{otherwise.}
  \end{cases}\label{eq:dom}
\end{align}
Any feasible allocation has to satisfy the \emph{occurrence}
constraints
\begin{align}
  &\sum_{s\in S} x_{si} = 1 & \forall\,  i\in N, \label{eq:fa1}
\end{align}
which ensure that every task is allocated to a single station, the
\emph{precedence} constraints
\begin{align}
  &x_{tj} \leq \sum_{s\in S\mid s\leq t} x_{si} & \forall\,  i\in N,j\in F_i, t\in S,\label{pr:bw}
\end{align}
and the \emph{nondivisibility} constraints
\begin{align}
  &x_{si}\in\{0,1\} & \forall\,  i\in N, s\in S\label{eq:fa3}\mathrm{.}
\end{align}
These constraints have been first proposed by \textcite{Bowman/1960}
and their above form is due to \textcite{White/1961}. For a given
cycle time $c$ and an upper bound $\overline m$ on the number of
stations, the SALBP-1 can be formulated as
 \minproblem[(BW1-1)\quad]{\sum_{s\in S} y_s,\label{BB1:1}}{
   \sum_{i\in N} t_i x_{si} \leq c y_s & \forall\,  s\in S,\label{BB1:2}\\
   & \text{Equations \eqref{eq:fa1}--\eqref{eq:fa3},}\label{BB1:3}\\
   & y_s\in\{0,1\} & \forall\, s\in S\textrm{,}\label{BB1:4}}
 where the variables $y_s$ indicate the usage of station $s\in S$. For
 the SALBP-2 the cycle time $c$ is variable and the number of stations
 $m$ is fixed. It can be formulated as
 \minproblem[(BW1-2)\quad]{c,\label{BB2:1}}{
   \sum_{i\in N} t_i x_{si} \leq c & \forall\,  s\in S,\label{BB2:2}\\
   & \text{Equations \eqref{eq:fa1}--\eqref{eq:fa3},}\label{BB2:3}\\
   & c\in\R\textrm{.}\label{BB2:4}}

These formulations are due to \textcite{Baybars/1986}. In the
following two subsections we present improvements of these models by
adding station limits and better formulations of the precedence
constraints.

\subsection{Station limits}
\label{sec:stationlimits}

For a given cycle time $c$ and a given number of stations $m$ one
often can derive bounds on the stations a task can be assigned to. For
task $i\in N$, let $E_i(c,m)$ be the earliest and $L_i(c,m)$ the
latest admissible station. (For the SALBP these bounds can be set, for
instance, to $E_i(c,m)=\ceilg{\sum_{j\mid j\leq i} t_j/c}$ and
$L_i(c,m) = m+1-\ceilg{\sum_{j\mid i\leq j} t_j/c}$.)  Then, we can
restrict the domain of the decision variables,
substituting~\eqref{eq:fa3} by
\begin{align}
  & x_{si}\in\{0,1\} & \forall i\in N, E_i(c,\overline m)\leq s\leq L_i(c,\overline m)\label{EL-1}
\end{align}
in the formulation of the SALBP-1 and by
\begin{align}
  & x_{si}\in\{0,1\} & \forall i\in N, E_i(\overline c,m)\leq s\in L_i(\overline c, m)\label{EL-2}
\end{align}
in the formulation of the SALBP-2, where $\overline m$ is an upper
bound on the number of stations, and $\overline c$ an upper bound on
the cycle time~\parencite{Patterson.Albracht/1975}.

The station bounds can be strengthened as
follows~\parencite{Pastor.Ferrer/2009}. In the case of the SALBP-1,
when taking into account the currently used stations,
\begin{align}
  & x_{L_i(c,s),i} \leq y_{s} & \forall\,  s\in S, i\in N\label{PF-1}
\end{align}
is valid, since if station $s$ is unused, we can assume that this also
holds for all later stations and therefore the latest possible station
is now $L_i(c,s-1)=L_i(c,s)-1$.

For the SALBP-2 we can model the cycle time explicitly to obtain
better bounds on the stations. Let $\underline c$ be a lower bound on
the cycle time and let $C=[\underline c,\overline c]$ be the set of
admissible cycle times. Then we can represent the cycle time by
\begin{align}
  & c = \sum_{t\in C} tr_t, \label{RT:1}\\
  & \sum_{t\in C} r_t = 1,\\
  & r_t\in\{0,1\} & t\in C\mathrm{.}\label{RT:3}
\end{align}
This allows to add the following inequalities to the formulation
of the SALBP-2:
\begin{align}
  & x_{ei} \leq 1-\sum_{t\in C\mid e<E_i(t,m)} r_t & \forall\,  i\in N, E_i(\overline c,m)\leq e<E_i(\underline c,m)\mathrm{,}\label{PF-2:1}\\
  & x_{li} \leq 1-\sum_{t\in C\mid L_i(t,m)<l} r_t & \forall\,  i\in N, L_i(\underline c,m)<l\leq L_i(\overline c,m)\label{PF-2:2}\mathrm{.}
\end{align}
These inequalities are easily seen to be valid. Suppose, for example,
that $\sum_{t\in C\mid e<E_i(t,m)} r_t=1$ in
equation~\eqref{PF-2:1}. Then station $e$ comes before the earliest
possible station for task $i$ considering the current cycle time, and
therefore $x_{ei}=0$. Similarly, if $\sum_{t\in C\mid L_i(t,m)<l}
r_t=1$ in equation~\eqref{PF-2:2} station $l$ comes after the latest
possible station for task $i$ given the current cycle time, and
therefore $x_{li}=0$.


\subsection{Alternative formulations of the precedence constraints}

\textcite{Patterson.Albracht/1975} have proposed to formulate the
precedence constraints as
\begin{align}
  &\sum_{s\in S} s x_{si} \leq \sum_{s\in S} s x_{sj} & \forall\,  i\in N, j\in F_i\label{pr:pa}\mathrm{.}
\end{align}
\textcite{Thangevelu.Shetty/1971} give the alternative formulation
\begin{align}
  & \sum_{s\in S} (m-s+1)(x_{si}-x_{sj}) \geq 0 & \forall\,  i\in N,j\in F_i\mathrm{.}\label{pr:ts}
\end{align}

Observe that both formulations are equivalent, since, by the
occurrence constraints we have
\begin{align*}
  \sum_{s\in S} (m-s)x_{si}+\sum_{s\in S} sx_{si}=m\mathrm{.}
\end{align*}

\section{Improved models for the SALBP}
\label{sec:improved}

\subsection{Station limits}
\label{sec:stationslimits}

We can strengthen the station limits proposed by
\textcite{Pastor.Ferrer/2009} as follows. In constraint \eqref{PF-1},
when a task cannot be assigned to station $L_i(c,s)$ it also cannot be
assigned to any later station. This justifies
\begin{align}
  & \sum_{u\in S\mid u\geq L_i(c,s)} x_{ui} \leq y_s & \forall\,  s\in S,i\in N\label{SPF-1}\mathrm{.}
\end{align}

Similarly, in constraints \eqref{PF-2:1} and \eqref{PF-2:2} when a
task cannot be assigned to a station earlier than $e$ this holds also
for stations preceding $e$, and when a task cannot be assigned later
than station $l$, this holds also for stations following
$l$. Therefore we can strengthen these constraints to
\begin{align}
  & \sum_{u\in S\mid u\leq e} x_{ui} \leq 1-\sum_{t\in C\mid e<E_i(t,m)} r_t & \forall\,  i\in N, E_i(\overline c,m)\leq e<E_i(\underline c,m)\label{SPF-2:1}\mathrm{,}
\end{align} and
\begin{align}
  & \sum_{u\in S\mid u\geq l} x_{ui} \leq 1-\sum_{t\in C\mid L_i(t,m)<l} r_t & \forall\,  i\in N, L_i(\underline c,m)<l\leq L_i(\overline c,m)\label{SPF-2:2}\mathrm{.}
\end{align}

\subsection{Precedence constraints}

We propose the following improved formulation of the precedence
constraints:
\begin{align}
  \sum_{s\in S\mid s\leq k} x_{si} & \geq \sum_{s\in S\mid s\leq k} x_{sj} & \forall\,  i\in N,j\in F_i, k\in S\mathrm{.}\label{pr:rc2}
\end{align}

The two following propositions state the validity and the theoretical
strength of these new constraints.

\begin{proposition}
  Constraints~\eqref{pr:rc2} are valid for BW1-1 and BW1-2.
\end{proposition}
\begin{proof}
  If $\sum_{s\in S\mid s\leq k} x_{si}=1$ the constraint is trivially
  satisfied, since $\sum_{s\in S} x_{sj}=1$. Otherwise, task $i$ is
  executed on station $k+1$ or later.  But since $j\in F_i$ task $j$
  cannot be executed on a station preceding station $k+1$, i.e.,
  $\sum_{s\in S\mid s\leq k} x_{sj}=0$.
\end{proof}

\begin{proposition}
  \label{pr:dom1}
  Inequalities~\eqref{pr:rc2} strictly dominate
  inequalities~\eqref{pr:pa} and
  \eqref{pr:bw}. Inequalities~\eqref{pr:pa} and \eqref{pr:bw} are
  incomparable.
\end{proposition}
\begin{proof} Patterson and Albracht's inequalities~\eqref{pr:pa} and
  the equivalent inequalities~\eqref{pr:ts} of Thangavelu and Shetty
  are aggregated versions of inequalities~\eqref{pr:rc2}. Indeed,
  summing up inequalities~\eqref{pr:rc2} for $k\in S$, we have, for
  all $i\in N, j\in F_i$,
\begin{align*}
   \sum_{k\in S} \sum_{s\in S\mid s\leq k} x_{si}
     \geq \sum_{k\in S} \sum_{s\in S\mid s\leq k} x_{sj}
\end{align*}
which can be easily rewritten as Patterson and Albracht's
inequalities~\eqref{pr:pa} since $$\sum_{k\in S}\sum_{s\in S\mid s\leq
  k} x_{si} = \sum_{s\in S} (m-s+1)x_{si}\mathrm{.}$$

Moreover, inequalities~\eqref{pr:rc2} also imply Bowman's
inequalities~\eqref{pr:bw}. Indeed, for all $i\in N, j\in F_i$ and
$t\in S$, the smaller terms of both inequalities compare as
\begin{align*}
  x_{tj} \leq \sum_{s\in S\mid s\leq t} x_{sj}\mathrm{,}
\end{align*}
and since their larger terms are equal, inequalities \eqref{pr:rc2} are lifted versions of~\eqref{pr:bw}.

The strictly dominance of inequalities \eqref{pr:rc2} over~\eqref{pr:pa} and~\eqref{pr:bw}, and the incomparability of the latter two can be seen by means of an example. Let $N=\{a,b\}$, $a\leq b$ and $m=3$. It is easy to verify
that the fractional solution $x_{1a}=x_{2a}=1/2$, $x_{1b}=3/4$, and
$x_{3b}=1/4$ satisfies~\eqref{pr:pa}, but neither \eqref{pr:rc2}, nor
\eqref{pr:bw}, and the fractional solution
$x_{1a}=x_{1b}=x_{2b}=x_{3a}=1/2$ satisfies~\eqref{pr:bw}, but neither
\eqref{pr:rc2}, nor \eqref{pr:pa}.
\end{proof}

The results of proposition~\ref{pr:dom1} are summarized in Figure~\ref{fig:relationships}, which
shows the relationships between the models.

\begin{figure}
  \centering
  \begin{tikzpicture}[scale=1.0,every node/.style={draw=black}]
    \node (rc) at (2,4.5) { NF };
    \node (pa) at (3,3) { PA };
    \node (bw) at (1,3) { BW };
    \node (ts) at (4.5,3) { TS };
    \draw[<->] (pa) -- (ts);
    \draw[->] (pa) -- (rc);
    \draw[->] (bw) -- (rc);
    \draw[dotted] (bw) -- (pa);

    \node (a1) at (7,4.5) { A };
    \node (b1) at (8,4.5) { B };
    \draw[->] (b1) -- (a1);
    \node[draw=none,anchor=west] at (9,4.5) { A is better than B };
    \node (a2) at (7,3.75) { A };
    \node (b2) at (8,3.75) { B };
    \draw[dotted] (b2) -- (a2);
    \node[draw=none,anchor=west] at (9,3.75) { A and B are incomparable };
    \node (a3) at (7,3.0) { A };
    \node (b3) at (8,3.0) { B };
    \draw[<->] (b3) -- (a3);
    \node[draw=none,anchor=west] at (9,3.0) { A and B are equivalent };
  \end{tikzpicture}
  \caption{Relationships between models of the SALBP with different
    formulations of precedence constraints. They are valid for the
    SALBP-1 as well as the SALBP-2 with the same station limits.}
  \label{fig:relationships}
\end{figure}
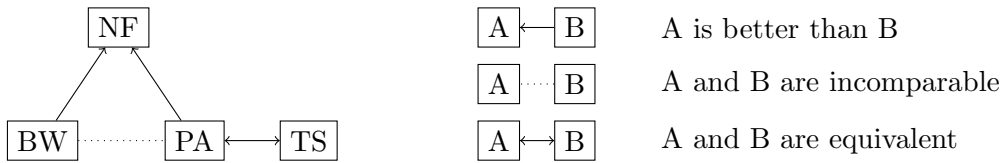

Tables~\ref{tab:1} and \ref{tab:2} give a summary of models with
different precedence constraints and station limits for the SALBP-1
and the SALBP-2 and their number of variables and restrictions.  We
obtain a different model for each combination of the precedence
constraints, the stations limits and the base model. For example,
BW2-1 denotes the formulation of the SALBP of type 1, using the
precedence constraints of Bowman~\eqref{pr:bw}, and station limits
\eqref{EL-1}. Note that equations~\eqref{EL-1} and \eqref{EL-2} do not
increase the number of restrictions, but reduce the number of
variables.  In both tables $o$ denotes the \emph{order strength} of
the instance, i.e., the fraction of the at most $\binom{n}{2}$
precedence relations present in the instance. The order strength of an
instance is at most $1$.

\begin{table}
  \centering
  \caption{Different formulations of the SALBP-1.\label{tab:1}}

  \medskip
  \begin{tabular}{lllllllllll}
    \toprule
    \multicolumn{4}{l}{Base model}                                                                                                        \\    
    Suffix                                     & Equations      & \#Var.               & \#Res.                                               \\
    \colrule
    -1                                         & \eqref{eq:fa1},\eqref{eq:fa3},\eqref{BB1:1},\eqref{BB1:2},\eqref{BB1:4}      & $\overline m(n+1)$ & $\overline m+n$                                    \\
    \colrule
    \multicolumn{3}{l}{Precedence constraints} & \multicolumn{3}{l}{Station limits}                                                       \\
    Name                                       & Equations      & \#Res.               & Name & Equations                  & \#Res.           \\ 
    \colrule
    PA                                         & \eqref{pr:pa}  & $on^2$             & 1    & -                          & -              \\
    BW                                         & \eqref{pr:bw}  & $on^2\overline m$  & 2    & \eqref{EL-1}               & -              \\
    TS                                         & \eqref{pr:ts}  & $on^2$             & 3    & \eqref{EL-1},\eqref{PF-1}  & $\overline mn$ \\
    NF                                         & \eqref{pr:rc2} & $on^2\overline m$  & 4    & \eqref{EL-1},\eqref{SPF-1} & $\overline mn$ \\
    \botrule
  \end{tabular}
\end{table}

\begin{table}
  \centering
  \caption{Different formulations of the SALBP-2.\label{tab:2}}

  \medskip
  \begin{tabular}{lllllllllll}
    \toprule
    \multicolumn{4}{l}{Base model}                                                                                                                                           \\
    Suffix                                     & Equations      & \#Var.      & \#Res.                                                                                           \\
    \colrule
    -2                                         & \eqref{eq:fa1},\eqref{eq:fa3},\eqref{BB2:1},\eqref{BB2:2},\eqref{BB2:4}  & $nm+1$    & $m+n$                                                                                          \\[0.25cm]
    \colrule
    \multicolumn{3}{l}{Precedence constraints} & \multicolumn{3}{l}{Station limits}                                                                                          \\ 
    Name                                       & Equations      & \#Res.      & Name & Equations                                                              & \#Res.           \\
    \colrule
    PA                                         & \eqref{pr:pa}  & $on^2$    & 1    & -                                                                      & -              \\
    BW                                         & \eqref{pr:bw}  & $on^2m$   & 2    & \eqref{EL-2}                                                           & -              \\
    TS                                         & \eqref{pr:ts}  & $on^2$    & 3    & \eqref{EL-2},\eqref{RT:1}-\eqref{PF-2:2}                               & $mn$           \\
    NF                                         & \eqref{pr:rc2} & $on^2m$   & 4    & \eqref{EL-2},\eqref{RT:1}-\eqref{RT:3},\eqref{SPF-2:1},\eqref{SPF-2:2} & $\overline mn$ \\
    \botrule
  \end{tabular}
\end{table}

\section{Application to related problems}
\label{sec:applications}

\subsection{An improved model for the UALBP-1}
\label{sec:ualbp1}

U-shaped assembly lines are an alternative to the traditional linear
layout, where tasks first pass all stations in the forward direction,
and then pass them again in the backward direction, in form of an
U. Therefore, the tasks assignable to a station include, besides tasks
whose predecessors have been assigned to a preceding station, also the
tasks whose successors have been assigned to a preceding station. This
added flexibility can improve the line's balance or reduce the number
of required stations. The problem of optimally allocating tasks to
stations is known as the U-line balancing problem
(UALBP). \textcite{Urban/1998} has proposed an integer linear program
for solving the UALBP-1 which is often used in studies of the UALBP
(e.g.~in~\parencite{Aase.etal/2003,Agpak.etal/2011,Chiang.etal/2004,Erin/2007,Gokcen.Agpak/2004,Kara.Tekin/2009}).

Introducing decision variables
\begin{align}
  \label{eq:dv}
  & x_{si} =
  \begin{cases}
    1 & \text{if task $i\in N$ is executed on station $s\in S$ in the forward pass}\\
    0 & \text{otherwise,}
  \end{cases}\\
  \intertext{and}
  & w_{si} = 
  \begin{cases}
    1 & \text{if task $i\in N$ is executed on station $s\in S$ in the backward pass}\\
    0 & \text{otherwise,}
  \end{cases}
\end{align}
the UALBP-1 is solved by the integer linear program
\minproblem{\sum_{s\in S} y_s}{
    \sum_{s\in S} x_{si}+w_{si} = 1 & \forall\,  i\in N,\\
  & \sum_{i\in N} t_i(x_{si}+w_{si})\leq cy_s & \forall\,  s\in S,\displaybreak[0]\\
  & \sum_{s\in S} (\overline m-s+1)(x_{si}-x_{sj}) \geq 0 & \forall\,  i\in N,j\in F_i,\label{ua:pr1}\\
  & \sum_{s\in S} (\overline m-s+1)(w_{si}-w_{sj}) \geq 0 & \forall\,  i\in N,j\in P_i,\label{ua:pr2}\\
  & x_{si}\in\{0,1\},w_{si}\in\{0,1\},y_s\{0,1\} & \forall\,  s\in S,i\in N\textrm{.}}

Due to the U-shaped layout every task may be assigned to the last
station, and the station limits for the last station do not
apply. Therefore, Urban applies only the bound
\begin{align}
  \label{eq:sl}
  E_i(c,m)=\min\left\{ \ceilgg{\sum_{j:j\leq i} t_j/c},\ceilgg{\sum_{j:i\leq j} t_j/c}\right\}
\end{align}
on the earliest station.

The above model uses precedence constraints as proposed by
\textcite{Thangevelu.Shetty/1971} for the SALBP. It can therefore be
improved by substituting \eqref{ua:pr1} and \eqref{ua:pr2} by the
precedence constraints
\begin{align}
  & \sum_{s\in S\mid s\leq k} x_{si} \geq \sum_{s\in S\mid s\leq k} x_{sj} & \forall\, i\in N,j\in F_i,k\in S,\\
  & \sum_{s\in S\mid s\leq k} w_{si} \geq \sum_{s\in S\mid s\leq k} w_{sj} & \forall\, i\in N,j\in P_i,k\in S\mathrm{.}
\end{align}
Furthermore, the station limits may be applied separately to the
forward and backward pass
\begin{align}
  & x_{si}\in\{0,1\} & \forall i\in N, E_i(c,\overline m)\leq s\mathrm{,}\\
  & w_{si}\in\{0,1\} & \forall i\in N, L_i(c,\overline m)\leq s\mathrm{,}
\end{align}
where $E_i(c,m)=\ceilg{\sum_{j:j\leq i} t_j/c}$ and $L_i(c,m) =
\ceilg{\sum_{j:i\leq j} t_j/c}$.

\subsection{An improved model for the BPP-P}
\label{sec:bppp}

The bin packing problem with precedence constraints asks to pack a set
of items into the smallest number of bins of a fixed size, with the
additional restriction that an item cannot share a bin with one of its
predecessors or successors.  It has been studied recently
by~\textcite{DellAmico.etal/2012}, who propose a mathematical model,
lower bounds, as well as heuristic and exact algorithms.

The BPP-P can be seen as a variant of the SALBP-1 with strict
precedences. Therefore, the improvements for SALBP-1 can be applied to
the model for the BPP-P, substituting the precedence
constraints~\eqref{pr:rc2} by their strict variant
\begin{align}
  \sum_{s\in S\mid s<k} x_{si} & \geq \sum_{s\in S\mid s\leq k} x_{sj} & \forall\,  i\in N,j\in F_i, k\in S\mathrm{.}\label{pr:rc2s}
\end{align}

The improved constraints for the station bounds~\eqref{SPF-1} also
apply to the BPP-P. They are stronger in the BPP-P, because the bounds
on the earliest station $E_i$ and the latest station $L_i$ a task can
be assigned to, improves when taking the strict precedences into
account. In our experiments below we use the better among the limits
for the SALBP-1 given in Section~\ref{sec:stationlimits} and the
limits imposed by the longest chain of predecessors or successors for
the earliest and latest station, respectively, for each task.

\section{Computational experiments}
\label{sec:computationalstudy}

We empirically evaluated the performance of all formulations of the
SALBP-1 and the SALBP-2 presented in Tables~\ref{tab:1} and
\ref{tab:2} and the improved formulations of the UALBP-1 and the
BPP-P.

For the SALBP, we limited the comparison to the best known
formulations PA and BW from the literature and the new formulation NF
for the precedence constraints combined with all four sets of
equations for the station limits, for a total of $12$ SALBP-1 and
SALBP-2 formulations. For the UALBP-1 and the BPP-P we compare the
model as originally proposed with the theoretically best model. In
this latter case, we also compare the results of the new model with
the results obtained with the tailored branch-and-bound algorithm of
\textcite{DellAmico.etal/2012}. The detailed results reported in the
tables below are available online at
\url{http://www.inf.ufrgs.br/algopt/albp}.

\subsection{Results for SALBP-1 and SALBP-2}

The formulations for the SALBP have been tested on the standard
benchmark which contains $269$ instances of the SALBP-1 and $302$
instances of the SALBP-2.
All instances are available online~\parencite{ALB}. Currently the
optimal value is known for the $269$ SALBP-1 instances and all except
$14$ of the SALBP-2 instances. In the evaluations below solutions for
instances without a known optimum were considered optimal only if the
solver could prove so. When comparing solution values obtained by
different formulations, a result is considered better when the null
hypothesis of no improvement in solution values can be rejected at
significance level $p=0.05$. In all tests, the test statistic used is
a conservative, non-parametric paired sign test, where half of the
ties were assigned to each sample~\parencite{Dixon.Mood/1946}.

The experiments were performed on a PC with an Intel Core i7 CPU
running at $2.8$ GHz and $12$ GB of main memory. We used the solver
CPLEX 12.4 with standard options, except for a MIP optimality gap of
$10^{-5}$, running in deterministic mode with two threads for a
maximum time of $600$ seconds. The computation times reported are in
seconds of real time. Following \textcite{Pastor.Ferrer/2009} we use
in our experiments for the SALBP-1 the lower bound $\underline
m=\ceil{\sum_{i\in N} t_i/c}$ and the upper bound $\overline
m=\min\{2\underline m,|N|\}$ on the number of stations, and for the
SALBP-2 the lower bound $\underline c=\max\{\max_{i\in N}
t_i,\ceil{\sum_{i\in N} t_i/m}\}$ and the upper bound $\overline
c=2\underline c$ on the cycle time.

Table~\ref{tab:3} contains the results for the SALBP-1 and
Table~\ref{tab:4} for the SALBP-2. For each tested model, we report
the number of instances for which the branch-and-cut solver of CPLEX
found a provably optimal solution within the time limit (Proven) and
the average solution time for these instances (Time). For the
remaining instances the solver did not terminate within the time
limit. For these runs, we report the number of instances for which an
optimal solution was found, but could not be proven to be optimal
(Unprov), the number of instances for which a feasible, but not
optimal solution was found (Feas), and the number of instances for
which the solver was unable to find a feasible solution
(Infeas). Column ``Best'' presents the number of instances where the
formulation obtained the best value found over all $12$ formulations
of the same problem type. The last column gives the average solution
time for the instances that could be solved optimally with all
models. For the SALBP-1 this was the case for $160$ instances, and for
the SALBP-2 for $147$ instances.

\begin{table}
  \tbl{Comparison of formulations of the SALBP-1 on $269$
    classical benchmark problems~\parencite{Scholl/1993}.\label{tab:3}}
{\begin{tabular}{lrrrrrrrr}
\toprule
      & \multicolumn{3}{c}{Optimal} &      & \multicolumn{2}{c}{Sub-optimal}            \\
\crule{\cline{2-4}\cline{6-7}}
Model & Proven                      & Time & Unprov. &  & Feas. & Infeas. & Best & Time \\
\colrule                                                                                    
PA1   & 167                         & 25.5 & 24      &  & 43    & 35      & 208  & 22.2 \\
PA2   & 182                         & 22.0 & 17      &  & 35    & 35      & 216  & 5.9  \\
PA3   & 187                         & 26.8 & 16      &  & 24    & 42      & 208  & 8.9  \\ 
PA4   & 189                         & 25.4 & 14      &  & 22    & 44      & 209  & 6.6  \\
\colrule                                                                                    
BW1   & 187                         & 39.3 & 17      &  & 40    & 25      & 210  & 15.5 \\ 
BW2   & 194                         & 30.4 & 16      &  & 44    & 15      & 220  & 4.3  \\ 
BW3   & 196                         & 33.5 & 14      &  & 19    & 40      & 220  & 5.3  \\ 
BW4   & 196                         & 28.3 & 14      &  & 26    & 33      & 228  & 3.4  \\
\colrule                                                                                    
NF1   & 194                         & 22.7 & 17      &  & 57    & 1       & 248  & 3.5  \\ 
NF2   & 197                         & 20.5 & 17      &  & 54    & 1       & 254  & 3.1  \\ 
NF3   & 201                         & 21.5 & 12      &  & 55    & 1       & 252  & 2.3  \\ 
NF4   & 200                         & 20.7 & 16      &  & 49    & 4       & 259  & 3.8  \\
\botrule
\end{tabular}}

\end{table}

\begin{table}
  \tbl{Comparison of formulations of the SALBP-2 on $302$
    classical benchmark problems~\parencite{Scholl/1993}.\label{tab:4}}
{\begin{tabular}{lrrrrrrrr}
\toprule
      & \multicolumn{3}{c}{Optimal} &      & \multicolumn{2}{c}{Sub-optimal}            \\
\crule{\cline{2-4}\cline{6-7}}
Model & Proven                      & Time & Unprov. &  & Feas. & Infeas. & Best & Time \\
\colrule
PA1   & 162                         & 49.4 & 18      &  & 122   & 0       & 192  & 39.8 \\ 
PA2   & 173                         & 51.4 & 15      &  & 114   & 0       & 202  & 31.7 \\ 
PA3   & 187                         & 40.5 & 9       &  & 91    & 15      & 204  & 8.2  \\ 
PA4   & 187                         & 33.7 & 11      &  & 88    & 16      & 206  & 11.3 \\ 
\colrule                                                                               
BW1   & 188                         & 38.6 & 11      &  & 103   & 0       & 216  & 16.1 \\ 
BW2   & 185                         & 36.5 & 16      &  & 101   & 0       & 220  & 11.2 \\ 
BW3   & 192                         & 26.0 & 10      &  & 94    & 6       & 218  & 5.4  \\ 
BW4   & 187                         & 29.7 & 8       &  & 106   & 1       & 211  & 8.8  \\ 
\colrule                                                                               
NF1   & 187                         & 33.7 & 16      &  & 99    & 0       & 237  & 8.5  \\ 
NF2   & 191                         & 39.2 & 17      &  & 94    & 0       & 247  & 6.2  \\ 
NF3   & 200                         & 36.9 & 8       &  & 94    & 0       & 224  & 6.2  \\ 
NF4   & 196                         & 33.9 & 11      &  & 95    & 0       & 225  & 6.4  \\ 
\botrule
\end{tabular}}

\end{table}

For both SALBP types, the results show a clear tendency to find and
prove more optimal solutions with better station limits and precedence
constraints. However, station limits of type 3 and 4 tend to make it
more difficult to find feasible solutions, and consequently reduce the
number of best solutions found. Statistically, the solution values
obtained by formulations NF$n$ are significantly less than those
obtained by the corresponding formulations PA$n$ and BW$n$ for both
problem types (with the exception of NF3, which is only marginally
better than BW3 for the SALBP-2). There is neither a significant
difference of solution values between precedence constraints BW and
PA, nor between different station limits.

The solution times also tend to decrease with better constraints, but
the reduction again is less pronounced or non-existent for station
limits of type 3 and 4. Statistically, formulations NF$n$ solve the
instances in significantly shorter time than formulations PA$n$ and
BW$n$, and station limits of type $2$ are better than those of type
$1$, which is expected they since only reduce the number of variables.

In summary, we find that station limits of type 2 help to reduce the
solution time, but in general better precedence constraints are more
important for improving solutions and reducing the solution time
than the improved station limits. In a previous study
\textcite{Pastor.etal/2004} found no significant difference between
formulations PA and BW, but observe that formulation BW leads to
shorter solution times. Our results confirm this, but we find the
reduction in solution time only significant for the SALBP-2. This may
come from the difference between the used solvers (CPLEX 8.0 and CPLEX
12.4). In another study \textcite{Pastor.Ferrer/2009} find that the
dynamic station limits (PA3) increase the number of provably optimal
solutions over formulation PA2, which is corroborated by our
findings. The solution values, on the other hand, do not decrease
significantly, and the dynamic station limits make it more difficult
to find feasible solutions.

From a practical point of view one may prefer the formulations NF4 for
the SALBP-1 and NF2 for the SALBP-2 which achieve the largest number
of best solutions in a short time.


\subsection{Results for the UALBP-1}

We tested Urban's formulation of the UALBP-1 and the improved
formulation proposed in Section~\ref{sec:ualbp1} on the $269$
instances of the SALBP-1. The experimental settings were the same as
in the tests of the SALBP related above. Table~\ref{tab:ualbp} shows
the comparison of the two models.

As expected, the conclusions for the SALBP-1 also hold for the
UALBP-1. The improved model finds and proves more optimal solutions,
and finds more and significantly better solution values (for $p=0.05$)
in about the same time used by the original model.

\begin{table}
  \tbl{Comparison of the standard formulation and an improved
    formulation for the UALBP-1 on $269$
    classical benchmark problems~\parencite{Scholl/1993}.\label{tab:ualbp}  }
{\begin{tabular}{lrrrrrrrr}
\toprule
         & \multicolumn{3}{c}{Optimal} &      & \multicolumn{2}{c}{Sub-optimal}            \\
\crule{\cline{2-4}\cline{6-7}}
Model    & Proven                      & Time & Unprov. &  & Feas. & Infeas. & Best & Time \\
\colrule
Standard & 175                         & 32.3 & 21      &  & 64    & 9       & 226  & 26.2 \\
    NF4  & 190                         & 44.7 & 23      &  & 55    & 1       & 256  & 32.5 \\
\botrule
  \end{tabular}}
  
\end{table}

\subsection{Results for the BPP-P}

We finally tested the formulation of the BPP-P proposed in
Section~\ref{sec:bppp} and compared it to the results obtained by
\textcite{DellAmico.etal/2012}. These results are available
online \parencite{DellAmico.etal/2010} and have been obtained in an
environment similar to ours (a PC with a Pentium processor running at
$3$ GHz, and CPLEX 12). To be able to make a direct comparison we used
the same settings running the solver with only one thread and a time
limit of two hours.

The results can be seen in Table~\ref{tab:bppp}. It reports for each
group of instances the number of instances (Ins), the results for the
Branch-and-bound algorithm as well as the model proposed by
\textcite{DellAmico.etal/2012}, and for the model proposed here. For
each approach the table gives the average relative deviation from the
best known lower bound (Dev), the average solution time (Time), and
the number of instances which could not be solved within the time
limit (Un). Such instances contribute with the time limit of $7200$s
to the average execution time. Note that the relative deviations
slightly differ from those reported in \textcite{DellAmico.etal/2010},
since we updated the lower bounds according to our results.

The new formulation drastically improves the results obtained by the
basic model. The number of instances that could not be solved in two
hours reduces from $81$ to $22$, and the average execution time is a
factor of $50$ faster. The best model is competitive with the
Branch-and-bound algorithm specifically designed for this problem:
although it solves $19$ problems less, it obtains a comparable
relative deviation on the remaining instances, and is able to solve
them a factor of almost four faster. Again, the conclusion for the
SALBP holds also for the BPP-P, and the improved constraints seem to
be even more effective for strict precedences.

\begin{table}
  \tbl{Comparison of the results of a Branch-and-bound algorithm
    and a formulation proposed by \textcite{DellAmico.etal/2012} (BW1)
    to the improved formulation NF4 of the BPP-P.\label{tab:bppp}}
  {\begin{tabular}{lrrrrr@{}rrrr@{}rrrr}
    \toprule
    &&&\multicolumn{3}{c}{Branch-and-bound}&&\multicolumn{3}{c}{BW1}&&\multicolumn{3}{c}{NF4}\\
    \crule{\cline{4-6}\cline{8-10}\cline{12-14}}
    Name & Ins. && Dev. & Time & Un. && Dev. & Time & Un. && Dev. & Time & Un.\\
    \crule{\cline{1-2}\cline{4-6}\cline{8-10}\cline{12-14}}

Arcus1    & 16  &  & 0.00 & 5.69   & - &  & 0.00  & 11.09   & -  &  & 0.00 & 0.07    & -  \\
Arcus2    & 17  &  & 0.00 & 211.49 & - &  & 1.70  & 2192.00 & 5  &  & 0.00 & 1,76    & -  \\
Barthold  & 8   &  & 0.00 & 4.41   & - &  & 0.00  & 80.66   & -  &  & 0.00 & 0.18    & -  \\
Barthol2  & 27  &  & 0.18 & 575.35 & 2 &  & 8.68  & 7200.00 & 27 &  & 1.89 & 5333.87 & 20 \\
Bowman    & 1   &  & 0.00 & 0.05   & - &  & 0.00  & 0.04    & -  &  & 0.00 & 0.00    & -  \\
Buxey     & 7   &  & 0.00 & 171.94 & - &  & 0.00  & 0.40    & -  &  & 0.00 & 0.02    & -  \\
Gunther   & 7   &  & 0.00 & 114.83 & - &  & 0.00  & 0.54    & -  &  & 0.00 & 0.02    & -  \\
Hahn      & 5   &  & 0.00 & 0.22   & - &  & 0.00  & 0.20    & -  &  & 0.00 & 0.01    & -  \\
Heskiaoff & 6   &  & 0.00 & 0.22   & - &  & 0.00  & 0.20    & -  &  & 0.00 & 0.01    & -  \\
Jackson   & 6   &  & 0.00 & 0.06   & - &  & 0.00  & 0.05    & -  &  & 0.00 & 0.00    & -  \\
Jaeschke  & 5   &  & 0.00 & 0.03   & - &  & 0.00  & 0.03    & -  &  & 0.00 & 0.00    & -  \\
Kilbridge & 10  &  & 0.00 & 0.58   & - &  & 0.00  & 1.15    & -  &  & 0.00 & 0.02    & -  \\
Lutz1     & 6   &  & 0.00 & 0.13   & - &  & 0.00  & 0.11    & -  &  & 0.00 & 0.00    & -  \\
Lutz2     & 11  &  & 0.00 & 106.99 & - &  & 6.63  & 3736.89 & 5  &  & 0.00 & 1.35    & -  \\
Lutz3     & 12  &  & 0.00 & 1.66   & - &  & 0.00  & 0.80    & -  &  & 0.00 & 0.06    & -  \\
Mansoor   & 3   &  & 0.00 & 0.08   & - &  & 0.00  & 0.04    & -  &  & 0.00 & 0.00    & -  \\
Mertens   & 6   &  & 0.00 & 0.03   & - &  & 0.00  & 0.02    & -  &  & 0.00 & 0.00    & -  \\
Mitchell  & 6   &  & 0.00 & 0.13   & - &  & 0.00  & 0.10    & -  &  & 0.00 & 0.00    & -  \\
Mukherje  & 13  &  & 0.00 & 76.92  & - &  & 0.61  & 1154.17 & 2  &  & 0.00 & 2.30    & -  \\
Roszieg   & 6   &  & 0.00 & 66.86  & - &  & 0.00  & 0.13    & -  &  & 0.00 & 0.00    & -  \\
Sawyer    & 9   &  & 0.00 & 134.48 & - &  & 0.00  & 5.41    & -  &  & 0.00 & 0.04    & -  \\
Scholl    & 26  &  & 0.00 & 114.39 & - &  & 31.52 & 6160.11 & 21 &  & 0.00 & 1.66    & -  \\
Tonge     & 16  &  & 0.00 & 15.74  & - &  & 1.51  & 1426.45 & 3  &  & 0.00 & 0.20    & -  \\
Warnecke  & 16  &  & 0.00 & 870.38 & 1 &  & 7.56  & 5513.81 & 12 &  & 0.00 & 109.73  & -  \\
Wee-Mag   & 24  &  & 0.00 & 2.24   & - &  & 1.29  & 2337.76 & 6  &  & 1.06 & 1089.12 & 2  \\
\crule{\cline{1-2}\cline{4-6}\cline{8-10}\cline{12-14}}
Tot./Avg. & 269 &  & 0.01 & 157.20 & 3 &  & 4.98  & 2289.93 & 81 &  & 0.12 & 40.06   & 22 \\
\botrule
  \end{tabular}}
\end{table}

\section{Conclusions}
\label{sec:conclusions}

We have proposed an improved formulation of the precedence constraints
and the station limits for the SALBP, and shown that they
theoretically dominate other constraints proposed in the
literature. They are applicable to related assembly line balancing
problems with similar constraints. Comparing the new and existing
models, we have provided a classification of the relationships between
models using impulse variable used in the literature.

Computational experiments confirm the theoretical comparison. The
proposed precedence constraints can improve upon the constraints of
\textcite{Patterson.Albracht/1975} and \textcite{Bowman/1960}, finding
and proving the optimality of more solutions, and finding more best
values. A conservative statistical test shows that the improvement of
the solution value is significant.

Two case studies on the UALBP-1 and the BPP-P indicate that the
conclusions for the SALBP also apply to related problems, which
further highlights the importance of the proposed models.

\bibliography{all}
\bibliographystyle{plainnat}

\end{document}